\documentclass[10pt, conference, letterpaper]{IEEEtran}
\IEEEoverridecommandlockouts
\usepackage{cite}
\usepackage{amsmath,amssymb,amsfonts}
\usepackage{algorithmic}
\usepackage{graphicx}
\usepackage[hidelinks]{hyperref}
\usepackage{textcomp}
\usepackage{xcolor}
\def\BibTeX{{\rm B\kern-.05em{\sc i\kern-.025em b}\kern-.08em
    T\kern-.1667em\lower.7ex\hbox{E}\kern-.125emX}}

\usepackage{amsthm}
\usepackage[small,full]{complexity}
\usepackage[ruled,linesnumbered]{algorithm2e}
\usepackage{stfloats}
\usepackage{multirow}

\newtheorem{theorem}{Theorem}

\newcommand{\defproblem}[3]{
\vspace{2.5mm}
\noindent
\fbox{
  \begin{minipage}{0.935\columnwidth}
  \begin{tabular*}{\columnwidth}{@{\extracolsep{\fill}}lr} #1

  \vspace{1.5mm}
  
  \end{tabular*}
  {\textbf{Input:}} #2 

  \vspace{1.5mm}
  
  {\textbf{Question:}} #3
  \end{minipage}
  }
\vspace{2.5mm}
}

\newcommand{\defoptproblem}[3]{
\vspace{2.5mm}
\noindent
\fbox{
  \begin{minipage}{0.935\columnwidth}
  \begin{tabular*}{\columnwidth}{@{\extracolsep{\fill}}lr} #1 

  \vspace{1.5mm}
  
  \end{tabular*}
  {\textbf{Input:}} #2

  \vspace{1.5mm}
  
  {\textbf{Objective:}} #3
  \end{minipage}
  }
\vspace{2.5mm}
}

\newcommand{\assign}{\textsc{CH-Assign}}
\renewcommand{\SL}{\textsc{SL-Makespan}}
\newcommand{\genSL}{\textsc{GenSL-Makespan}}

\begin{document}

\title{Makespan Minimization in Split Learning: 

From Theory to Practice}

\author{\IEEEauthorblockN{Robert Ganian$^{\star}$, Fionn {Mc~Inerney}$^{\dagger}$, and Dimitra Tsigkari$^{\dagger}$}
\IEEEauthorblockA{$^{\star}$Algorithms and Complexity Group, TU Wien, Vienna, Austria, rganian@gmail.com\\
$^{\dagger}$Telef\'{o}nica Scientific Research, Barcelona, Spain, \{fionn.mcinerney, dimitra.tsigkari\}@telefonica.com}
\thanks{All of the authors contributed equally to this work. It will appear in the proc. of \textbf{IEEE INFOCOM 2026}, and was funded by the Austrian Science Fund (FWF) [10.55776/Y1329 and 10.55776/COE12], the WWTF Vienna Science and Technology Fund (Project 10.47379/ICT22029), the Smart Networks and Services Joint Undertaking
(SNS JU) under the European Union’s Horizon Europe and innovation programme under Grant Agreement No 101139067 (ELASTIC), and the Horizon MSCA Postdoctoral Fellowship OPALS (grant agreement 101210495). Views and opinions expressed are however those of the authors only and do not necessarily reflect those of the EU. Neither the EU nor the granting authority can be held responsible for them. 

\centering \includegraphics[width=2.87cm,trim={0cm 0.53cm 0cm 0cm},clip]{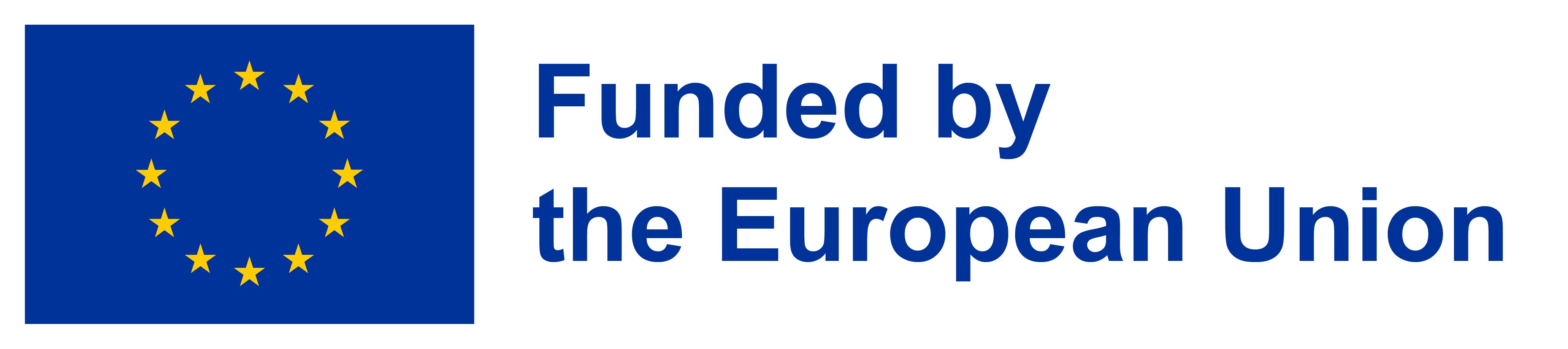}}}

\maketitle

\begin{abstract}
Split learning recently emerged as a solution for distributed machine learning with heterogeneous IoT devices, where clients can offload part of their training to computationally-powerful helpers.
The core challenge in split learning is to minimize the training time by jointly devising the client-helper assignment and the schedule of tasks at the helpers. 
We first study the model where each helper has a memory cardinality constraint on how many clients it may be assigned, which represents the case of homogeneous tasks. Through complexity theory, we rule out exact polynomial-time algorithms and approximation schemes even for highly restricted instances of this problem.
We complement these negative results with a non-trivial polynomial-time 5-approximation algorithm.
Building on this, we then focus on the more general heterogeneous task setting considered by Tirana et al.~[INFOCOM~2024], where helpers have memory capacity constraints and clients have variable memory costs.  
In this case, we prove that, unless $\P=\NP$, the problem cannot admit a polynomial-time approximation algorithm for any approximation factor.
However, by adapting our aforementioned 5-approximation algorithm, we develop a novel heuristic for the heterogeneous task setting and show that it outperforms heuristics from prior works through extensive experiments.
\end{abstract}

%\begin{IEEEkeywords}

%\end{IEEEkeywords}

\section{Introduction}

Split Learning~(SL) is a modern framework which allows devices~(\emph{clients}) with heterogeneous computing resources to collaboratively train a deep Neural Network~(NN) model~\cite{vepakomma2018split}. Relaxing the requirement of Federated Learning~(FL)~\cite{mcmahan2017communication} according to which the training takes place at the clients, SL allows clients to offload part of the training to a \emph{helper}---typically a server in the cloud or a base station in wireless communications. This allows heterogeneous clients which lack the necessary resources to train a model locally to still participate in distributed learning, potentially making use of underutilized network resources while accelerating the training process. For these reasons, the SL paradigm is of growing importance in Internet-of-Things~(IoT) networks and has numerous applications in 6G networks~\cite{samikwa2022ares, lin2024split}. 

\begin{figure}[!t]
	\centering  
	\includegraphics[width=0.9\columnwidth, trim={4.8cm 10cm 8.6cm 1.25cm},clip]{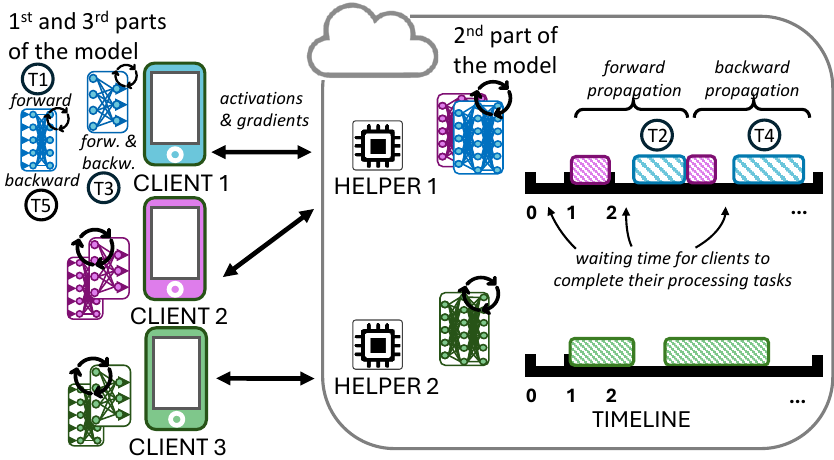} 
    \caption{An example of client-helper assignments and scheduling decisions. Processing tasks 1 to 5 (T1--T5) correspond to the model parts of Client 1.}
	\label{fig:setup}
\end{figure}

Based on the architecture of a NN model that is organized into layers, i.e., blocks of neurons, in SL, the model is  split into 3 parts. These parts are composed of consecutive layers and separated by \emph{cut layers}. Part-1 and part-3 are processed at the clients, while part-2 at one of the helpers. Clients retain data ownership when splitting this way, whereas a 2-part split requires the clients to share the labels of their data with the helper~\cite{vepakomma2018split}. In this work, we consider the SL training framework where the cut layers are predetermined, the clients train their model parts in parallel, and the helpers maintain a distinct copy of part-2 for each assigned client; at the end of each training round, the model parts are aggregated (e.g., through \texttt{FedAvg}~\cite{mcmahan2017communication}) at an aggregator or a helper that assumes this role. This variant is typically called SplitFedV1~\cite{thapa2022splitfed}; crucially, it has been shown to exhibit a similar accuracy to FL~\cite{gao2021evaluation}, and it has been extensively studied in the literature~\cite{wang2023coopfl,  han2024convergence, xia2025multisfl}.

The three model parts need to be processed consecutively during the forward propagation and in the inverse order during backward propagation. Therefore, for each client, the training process is naturally decomposed into \emph{five processing tasks}, as depicted in Fig.~\ref{fig:setup}. These tasks are: forward propagation of part-1 (T1), then part-2 (T2), forward and backward propagation of part-3 (T3), backward propagation of part-2 (T4), and then part-1 (T5).  The handover between the client and helper is achieved through the exchange of activations (during the forward pass) and gradients (during the backward pass). The described bidirectional communication may severely affect the total training time~(makespan), which is the maximum time needed to train the model across all clients. Given that achieving a short training time is crucial in a multitude of applications (including, e.g., Zero-Touch Networks~\cite{benzaid2020ai}), this creates \emph{a need for optimizing the SL operations}---and specifically for optimizing the makespan. This optimization is also important in view of scalability, which is one of the main challenges faced by the SL paradigm~\cite{gao2021evaluation,lyu2023scalable, hafi2024split}.

While minimizing the makespan in the setting of multiple compute nodes  is a classic optimization problem in scheduling and resource allocation, SL poses unique challenges. In particular, the limited capacity constraints at the helpers, the precedence constraints of the processing tasks~(as described above), and the distinct release times of the tasks do not allow us to directly apply known algorithmic results from the general scheduling literature.
The related works that formulated our problem of interest and several variants of it have provided heuristics for makespan minimization~\cite{wang2023coopfl,TTIC25,tirana2024workflow}---see the discussion in Section~\ref{subsec:related}. Apart from the weak \NP-hardness of the problem established in~\cite{TTIC25}, we lack a clear understanding of the computational bounds and approximability of SL makespan minimization. This is the first work that fills this gap, as we elaborate below.

\subsection{Our Contributions}
We first consider the complexity of computing a client-helper assignment and schedule that minimizes the training time---a problem which we denote \SL\ (see Section~\ref{sec:defs} for formal definitions of all the considered problems). As mentioned above, \SL\ is weakly \NP-hard, even if there are only two helpers~\cite{TTIC25}. We adapt the known reduction to establish that even when restricted to instances where each client is adjacent to each helper, the only tasks with non-zero processing times are the T2s, and every helper has the same processing capabilities, \SL\ is (1) strongly \NP-hard even if at most 3 clients may be assigned to each helper and (2) \W\textup{[1]}-hard\footnote{Under the well-established hypothesis that $\FPT\neq \W[1]$ (see, e.g.,~\cite{parambook}), \SL\ with input size $n$ cannot be solved in $f(I)\cdot n^{\mathcal{O}(1)}$ time for any computable function $f$.} parameterized by the number of helpers $I$ (Theorem~\ref{thm:sched-hard}).
Moreover, even if there is only a single helper, it can be observed from a result of~\cite{YHL04} that merely scheduling the tasks to minimize the makespan is strongly \NP-hard even if the processing times of the T1s and T5s are zero and the T2s and T4s have unit processing times (Theorem~\ref{thm:sched-hard2}). 

Given the inherent intractability of finding exact solutions for \SL, we turn to approximation.
Unfortunately, we establish that, unless $\P=\NP$, \SL\ cannot admit a polynomial-time approximation scheme (PTAS). In fact, via a reduction from $\text{R}~||~C_{\text{max}}$ to \SL, we prove an even stronger statement: unless $\P=\NP$, even when restricted to instances where each client is adjacent to each helper and the only tasks with non-zero processing times are the T2s, \SL\ cannot admit a polynomial-time (3/2)-approximation algorithm (Theorem~\ref{thm:inapprox}). Note that the inapproximability for $\text{R}~||~C_{\text{max}}$ cannot be lifted much further as it is well known that it admits a polynomial-time 2-approximation algorithm~\cite{LST90}, and closing this gap is a longstanding open question~\cite{BLMRS24,Svensson12}. 
We show that the same holds true for \SL. That is,
despite the accumulating evidence that \SL\ is hopelessly intractable even for highly restricted instances, as our \textbf{first main theoretical contribution}, we provide the first approximation algorithm for \SL, in particular, a polynomial-time 5-approximation algorithm (Theorem~\ref{thm:approx}). 

Building on this, we then consider the generalization of \SL\ introduced in~\cite{TTIC25,tirana2024workflow}, that we call \genSL, in which each helper has a memory capacity and needs to allocate part of its memory for the activation received from a client, where the amount of memory can differ for each client. Note that \SL\ can be seen as a particular case of \genSL\ where the amount of memory for each client is the same, and thus, can be considered to be of unit cost. As the problem of scheduling at a single helper is already hard  for \SL\ (Theorem~\ref{thm:sched-hard2}), we similarly show that the problem of determining whether there is a feasible client-helper assignment for \genSL\ (denoted \assign) is hard on its own. Specifically, we provide a polynomial-time reduction from $\text{P}~||~\text{C}_{\text{max}}$ to \assign, and thus, obtain hardness results analogous to those of \SL\ (Theorem~\ref{thm:assign-hard}). In particular, as it is strongly \NP-hard to decide if a feasible client-helper assignment exists, as our \textbf{second main theoretical contribution}, we establish that, unless $\P=\NP$, \assign\ cannot admit a polynomial-time approximation algorithm for any approximation factor. 

Nevertheless, as our \textbf{main practical contribution}, we suitably adapt our 5-approximation algorithm for \SL\ to obtain a heuristic for \genSL\ that we call \textbf{EquiDistributed} (\textbf{EquiD} for short). It is based on finding a feasible client-helper assignment by solving an integer program (e.g., via a solver), while the scheduling part is the same as in our 5-approximation algorithm.
In Section~\ref{sec:eval}, through experimental results\footnote{The code is available at https://github.com/ddd-ttt99/SL-makespan.} based on open-source and synthetic data, we then show that EquiD outperforms the other heuristics from the literature, and finds near-optimal solutions in a fraction of the time it takes a solver to find the optimal solution.

\subsection{Related Work} \label{subsec:related}
SL was first introduced in~\cite{vepakomma2018split} (and \cite{gupta2018distributed} in the same year) as a distributed learning method where the computational requirements for the participating devices are lower when compared to FL. Several variants of the framework have been considered since then, with the most prominent ones being the SplitFed~\cite{thapa2022splitfed} variants. Apart from works focusing on minimizing the makespan in SL (on which we elaborate below),  related work studies problems related to privacy and security~\cite{pasquini2021unleashing, thapa2022splitfed, zhang2023privacy}, convergence under heterogeneous data~\cite{han2024convergence, liao2024mergesfl, xia2025multisfl, tirana2025data},  cut layer selection~\cite{wang2021hivemind, lin2025hierarchical}, and energy-efficiency~\cite{samikwa2022ares, liu2022energy, kim2023bargaining} in the setting of SL. Finally, some works in the literature focus on a ``hybrid'' setting where clients may participate in the training either through FL or SL, e.g.,~\cite{liu2022wireless, liu2022novel, TTIC25}.

Minimizing the training makespan of SL has been identified as one of the main challenges of SL~\cite{hafi2024split}. The work in \cite{samikwa2022ares} formulates the problem of jointly minimizing the makespan and the energy consumption. However, it assumes that each helper serves a single client, which vastly simplifies client-helper assignments and eliminates the need for scheduling. The work in~\cite{wang2023coopfl} studies the problem of makespan minimization with a single helper and proposes a method that decides on the cut layer and the schedule at the helper. On the other hand, the authors in~\cite{tirana2024workflow} formulate the problem of makespan minimization in the presence of multiple helpers and propose two heuristics. The first one is based on a decomposition of the problem into subproblems, and the second one  finds the client-helper assignments through a load-balancing  algorithm, and employs a first-come-first-serve scheduling policy. 
Finally, there are works that minimize the training time by considering decisions like cut layer selection, radio spectrum allocation~(in wireless networks), and device clustering, e.g.,~\cite{wu2023split, lin2024efficient, huang2024decentralized}.

\section{Problem Formulations}\label{sec:defs}

We consider a system with a set $\mathcal{J}=\{1,\ldots,J\}$ of clients (devices) and a set $\mathcal{I}=\{1,\ldots,I\}$ of helpers~(located at the cloud or the network edge), which are connected  with  links~$\mathcal{E}$. These links may reflect compute configurations at the cloud or location-based policies at the edge. This naturally defines a bipartite graph $G=(\mathcal{J}, \mathcal{I}, \mathcal{E})$. 
While helpers usually have the same compute/memory capacity in related work, e.g.,~\cite{samikwa2022ares}, in the second part of the paper we follow the framework proposed in more recent works and
consider 
the generalized problem where the helpers are processors of various capabilities~(e.g., CPUs or GPUs). 
This generalization allows to capture realistic scenarios of high heterogeneity~\cite{weng2022mlaas}. In particular,  the memory capacity of helper $i\in \mathcal{I}$ is denoted by $M_i$ and it reflects its physical memory~(e.g., measured in Mbytes or Gbytes). Finally, the computing capacities of helpers can be described either by cycles per second or by the estimated time of the computing task (after profiling), as described below.

\subsection{The Steps of SL Training} 
In order to minimize the training makespan, we focus on minimizing the training time of a single batch of data for each client, as in previous work~\cite{wang2023coopfl, tirana2024workflow,TTIC25}. We note that in applications, the training process is naturally decomposed into multiple consecutive batch training processes---where the number of batches depends on how the clients' data is partitioned, the number of epochs and training rounds, etc.

As is common in the scheduling literature, we employ a time-slotted model with time intervals $S_t$, where $S_t=[t, t+1]$ for all integers $t\geq 0$ and the time slots are measured in, e.g.,~msec. For each batch of  data (of equal size) of client $j\in \mathcal{J}$, the following tasks take place:
\begin{itemize}
\item[\textbf{T1:}] Client $j$ performs forward propagation of part-1 and communicates the activations of the (predetermined) cut layer to the cloud in $r_j$ time slots.
\item[\textbf{T2:}] The assigned helper $i\in \mathcal{I}$ performs forward propagation of part-2 in $p_{ij}$ time slots. Note that this time depends on the processing capabilities of helper $i$, as well as the size of part-2 for client $j$ (which may vary for each client).
\item[\textbf{T3:}] Client $j$ receives the activations of part-2 from helper~$i$ and performs forward propagation of part-3. Client~$j$ then computes the loss function, starts the backward propagation of part-3~(i.e., computes gradients and model weights), 
and transmits the gradients to the cloud. These operations last $\ell_j$ time slots.
\item[\textbf{T4:}] Helper $i$ performs backward propagation of part-2 in $p_{ij}^\prime$ time slots.
\item[\textbf{T5:}] Client $j$ receives the gradients of part-2 and back-propagates them into part-1 in $r_j^\prime$ time slots. This completes the training of a batch.
\end{itemize}

The helper allocates part of its memory for the activations/gradients received from each assigned client as well as for the model weights of part-2 of each client. We denote by~$d_j$ the memory footprint of client~$j$ at the assigned helper. We distinguish two different cases of how the memory capacity constraints are modeled in the makespan minimization problem. First, these constraints are considered to be cardinality constraints, i.e., determining the maximum number of clients that can be served at each helper. 
Given that each helper has a limited memory capacity or other communication constraints, helper $i$ can serve up to $M_i$ clients during a batch update, for all $i\in \mathcal{I}$, which is equivalent to assuming $d_j=1$ for all $j\in \mathcal{J}$. The assignments of clients to helpers should respect this \emph{basic servicing constraint} that captures homogeneous scenarios in terms of processing tasks~(i.e., the same memory footprint for each client).
Later, in Section~\ref{sec:genSL}, we also consider a \emph{comprehensive servicing constraint} which allows clients to have distinct memory demands (i.e., distinct values of~$d_j$), in line with previous work~\cite{TTIC25,tirana2024workflow}. We assume that the aforementioned memory capacities $M_i$ and times $r_j$, $p_{ij}$, $\ell_j$, $p_{ij}'$, and $r_j'$ are all non-negative integers.

Let a client-helper assignment $Y: \mathcal{J} \rightarrow \mathcal{I}$ be a mapping that assigns clients to helpers and let $Z_Y(i):=\{j\mid Y(j)=i\}$ for all $i\in \mathcal{I}$.
A client-helper assignment $Y$ is \emph{feasible} if it complies with (a) the adjacency constraints of $G$, i.e., for all $j\in \mathcal{J}$, $(j,Y(j))\in \mathcal{E}$, and (b) the servicing constraints of the helpers, i.e., for all $i\in \mathcal{I}$, $\sum_{j\in Z_Y(i)} d_j \leq M_i$. 
We remark that the client-helper assignment cannot change midway through processing the client's batch~\cite{wang2023coopfl, tirana2024workflow}.

\subsection{Training Makespan of a Batch} 

In basic scheduling problems, a schedule is an assignment of jobs $\mathcal{J}$ to time intervals of machines $\mathcal{I}$ such that each job is assigned to one machine and no two jobs are assigned to the same time interval of the same machine (i.e., single-threaded processing). For each job $j\in \mathcal{J}$ assigned to machine $i\in \mathcal{I}$, the number of time intervals it needs to be assigned to is equal to its processing time $p_j$ in the case of identical parallel machines or $p_{ij}$ in the case of unrelated parallel machines. If preemption is allowed, then the assigned time intervals for a job do not have to be consecutive. If no graph and no machine capacity constraints are specified (such as in the definitions of $P~||~C_{\text{max}}$ and $R~||~C_{\text{max}}$ in Section~\ref{sec:SLhard}), then any job may be assigned to any machine. The makespan of a schedule is the time the last job is completed. 

Our setting is analogous to the one described above, except that the jobs are now the T2s and T4s of the clients, the machines are the helpers, there are graph and helper memory cardinality/capacity constraints, and preemption is allowed. We highlight that we do not make use of preemption in any of our results, but we allow it in order to be inline with the model of~\cite{TTIC25,tirana2024workflow}. Only the T2s and T4s of the clients must be assigned to time intervals of the helpers; note that the T2 and T4 of the same client must be assigned to the same helper and the T4 cannot be scheduled before the T2 is completed. Furthermore, there are the following additional constraints. Given a client $j$: (a)~the release date of its T2 is equal to $r_j$, which is the time needed to complete its T1 (i.e., T2 cannot start being processed before time~$r_j$) and (b)~there is a delay of $\ell_j$ time slots (the time needed to complete its T3) between the completion of its T2 and the time that its T4 may start. The completion time of the batch for client $j$, denoted by~$c_j$, is the completion time of its T5 which occurs~$r_j'$~time slots after the completion of its~T4. These requirements derive from the precedence relations of the tasks within the SL workflow.

We emphasize that since clients are only responsible for part-1 and part-3 of their local model (i.e., Tasks 1, 3, and~5), they process them as soon as they are available: Task 1 at time $0$, Task 3 after Task 2 is completed at the assigned helper, and Task 5 once Task 4 is completed. In principle, a natural assumption in distributed learning is that  devices are idle during training~\cite{google_fl}. Therefore, no schedules need to be defined for the tasks being processed at the clients.

The makespan of a batch is determined by the maximum training time among all clients. Therefore, our goal is to minimize the \emph{makespan} $\max_{j\in\mathcal{J}} c_j$.

We now have the necessary definitions and notation to formally define our considered problems. As is common when proving formal computational hardness results, we consider the decision variants of the problems at hand. We remark, however, that all of our algorithms apply to the optimization versions\footnote{We abuse notation and use the same name for both the decision and optimization versions of the considered problems.} and are constructive, i.e., they can output a feasible corresponding schedule along with the makespan that they compute. In the case of memory cardinality constraints for the helpers, we study the following problem.

\defproblem{\textsc{\SL}}{A bipartite graph $G=(\mathcal{J}, \mathcal{I}, \mathcal{E})$, the memory capacities $M_1, \ldots, M_I$ of the helpers, the memory demands $d_1=\dots=d_J=1$ of the clients, the times $r_j$, $p_{ij}$, $\ell_j$, $p_{ij}^\prime$, and $r_j^\prime$, for all $i\in \mathcal{I}$, and $j\in\mathcal{J}$, and a positive integer $k$.}
{Is there a feasible client-helper assignment and schedule with makespan at most~$k$?}

In the more general case of memory capacity constraints---which capture scenarios of high heterogeneity---we consider the following generalization of \SL\ introduced and studied in~\cite{TTIC25,tirana2024workflow}.

\defproblem{\textsc{Generalized SL Makespan} (\genSL)}{A bipartite graph $G=(\mathcal{J}, \mathcal{I}, \mathcal{E})$, the memory capacities $M_1, \ldots, M_I$ of the helpers, the memory demands $d_1, \ldots, d_J$ of the clients, the times $r_j$, $p_{ij}$, $\ell_j$, $p_{ij}^\prime$, and $r_j^\prime$, for all $i\in \mathcal{I}$, and $j\in\mathcal{J}$, and a positive integer $k$.}
{Is there a feasible client-helper assignment and schedule with makespan at most $k$?}

Both of the optimization versions of \SL\ and \genSL\ have the same input as above, but instead ask for a feasible client-helper assignment and schedule with minimum makespan.
In the case of \genSL, to establish that merely determining whether a feasible client-helper assignment exists is strongly \NP-hard, we define the following restriction of \genSL.

\defproblem{\textsc{Client-Helper Assignment} (\assign)}{A bipartite graph $G=(\mathcal{J}, \mathcal{I},\mathcal{E})$, the memory capacities $M_1, \ldots, M_I$ of the helpers, and the memory demands $d_1,\ldots,d_J$ of the clients.}
{Is there a feasible client-helper assignment?}

\section{\SL: Hardness and Approximability}\label{sec:SLhard}

In this section, we study the computational complexity of solving \SL\ exactly and approximately. Some of our formal computational hardness results are established via polynomial-time reductions from two very well-known scheduling problems in which the goal is to minimize the makespan: (a) the problem of scheduling on identical parallel machines, denoted by $P~||~C_{\text{max}}$, and (b) the problem of scheduling on unrelated parallel machines, denoted by $R~||~C_{\text{max}}$. These two problems are defined as follows.

\defproblem{$P~||~C_{\text{max}}$}
{A set of jobs $\mathcal{J}$, a set of identical parallel machines $\mathcal{I}$, a non-negative integer job processing time $p_j$ for all $j\in \mathcal{J}$, and a positive integer $k$.}
{Is there a schedule with makespan at most $k$?}

\defproblem{$R~||~C_{\text{max}}$}
{A set of jobs $\mathcal{J}$, a set of unrelated parallel machines $\mathcal{I}$, a non-negative integer job processing time $p_{ij}$ for all $i\in \mathcal{I}$ and $j\in \mathcal{J}$, and a positive integer $k$.}
{Is there a schedule with makespan at most $k$?}

We first extend the weak \NP-hardness result in the case of two helpers~\cite{TTIC25} by proving that severely restricted instances of \SL\ remain hard.

\begin{theorem}\label{thm:sched-hard}    
    Even when restricted to instances such that $G$ is a complete bipartite graph and $r_j=\ell_j=p_{ij}'=r_j'=0$ and $p_{ij}=p_{i'j}$ for all $i,i'\in \mathcal{I}$ and $j\in \mathcal{J}$, \SL\ is \textup{(1)} strongly \NP-hard even if $M_1=\dots=M_I=3$ and \textup{(2)} \W[1]-hard parameterized by $I$.
\end{theorem}

\begin{proof}
    If $G$ is a complete bipartite graph and $r_j=\ell_j=p_{ij}'=r_j'=0$ and $p_{ij}=p_{i'j}$ for all $i,i'\in \mathcal{I}$ and $j\in \mathcal{J}$, then \SL\ is equivalent to $P~||~C_{\text{max}}$, where the clients and helpers in \SL\ correspond to the jobs and machines, respectively, in $P~||~C_{\text{max}}$. The statement of the theorem then holds since $P~||~C_{\text{max}}$ is known to be (1) strongly \NP-hard even if exactly 3 jobs must be assigned to each machine (this particular case of $P~||~C_{\text{max}}$ is known as the {\sc 3-Partition} problem)~\cite{GJ79} and (2) \W[1]-hard parameterized by the number of machines~\cite{JKMS13,KK18}.
\end{proof}

The hardness from the previous theorem stems from deciding the client-helper assignment. However, if this assignment is enforced (or when there is a single helper), the scheduling is also computationally intractable---even when highly restricted. In fact, this is a direct result of some restricted instances of \SL\ being equivalent to a known scheduling problem from the literature.

\begin{theorem}{\textup{\cite[Thm. 24]{YHL04}}}\label{thm:sched-hard2}
    Even when restricted to instances such that $I=1$, $G$ is a complete bipartite graph, and $r_j=r'_j=0$ and $p_{1j}=p_{1j}'=1$ for all $j\in \mathcal{J}$, \SL\ is strongly \NP-hard.
\end{theorem}

Since solving \SL\ exactly is inherently intractable, a natural next step is to instead aim for approximate solutions. Unfortunately, we establish that it is also hard to solve \SL\ approximately, even in very restricted instances. In particular, \SL\ cannot admit a PTAS unless $\P=\NP$. In fact, an even stronger statement holds:

\begin{theorem}\label{thm:inapprox}
    Unless $\P=\NP$, even when restricted to instances such that $G$ is a complete bipartite graph and $r_j=\ell_j=p_{ij}'=r_j'=0$ for all $i\in I$ and $j\in J$, \SL\ cannot admit a polynomial-time (3/2)-approximation algorithm.
\end{theorem}

\begin{proof}
         If $G$ is a complete bipartite graph and $r_j=\ell_j=p_{ij}'=r_j'=0$ for all $i\in I$ and $j\in J$, then \SL\ is equivalent to $R~||~C_{\text{max}}$, where the clients and helpers in \SL\ correspond to the jobs and machines, respectively, in $R~||~C_{\text{max}}$. The statement of the theorem then holds since $R~||~C_{\text{max}}$ cannot admit a polynomial-time $(3/2)$-approximation algorithm unless $\P=\NP$~\cite{LST90}.
\end{proof}

\begin{algorithm*}[t]
\NoCaptionOfAlgo

\SetAlgoLined
\SetKwInOut{Input}{Input}
\SetKw{Return}{Return}
\SetAlgoLined

\Input{$G=(\mathcal{J}, \mathcal{I}, \mathcal{E})$, $M_1, \ldots, M_I$, $r_j$, $p_{ij}$, $\ell_j$, $p_{ij}'$, and $r_j'$, for all $i\in \mathcal{I}$ and $j\in\mathcal{J}$.}
Apply the polynomial-time 2-approximation algorithm of~\cite[Thm.~2.3]{SS18} for {\sc GAPcc} using $G$, $M_i$, and $p^{\star}_{ij}:=p_{ij}+p_{ij}'$ for all $i\in \mathcal{I}$ and $j\in \mathcal{J}$ as input. This returns a feasible client-helper assignment $Y: \mathcal{J} \rightarrow \mathcal{I}$.

\For(\tcp*[f]{\small Initialize time  each $T4$ can be processed}){$j=1:J$}
  {$w_j:=\infty$}

 \For(\tcp*[f]{\small Schedule clients in $Z_Y(i)$ for each helper $i$}){$i=1:I$}
 {Sort clients in $Z_Y(i)$ in decreasing order w.r.t.~$\ell_j$. Let $Q$ be the ordered set.

 Sort clients in $Z_Y(i)$ in decreasing order w.r.t.~$r_j'$. Let $Q'$ be the ordered set.

 $t:=0$ \tcp*[f]{\small Initialize time}

 \While(\tcp*[f]{\small Schedule $T2$s and $T4$s of clients in $Z_Y(i)$}){$|Q|\neq 0$ OR $|Q'|\neq 0$}
 {$t=\max \{t,\min_{j\in Q,j'\in Q'}\{r_j,w_{j'}\}$\} \tcp*[f]{\small Adapt $t$ if no $T2$ or $T4$ can be processed}

 \If{$|Q|\neq 0$ AND $t\geq \min_{j\in Q}r_j$}
 {Assign $T2$ of the client $j$ with the smallest index in $Q$ such that $r_j\leq t$ to time intervals $S_t$ through $S_{t+p_{ij}}$ on helper $i$.

 $Q=Q\setminus \{j\}$ \tcp*[f]{\small Update $Q$ to remove $T2$ of client $j$}

 $t=t+p_{ij}$ \tcp*[f]{\small Update $t$ to time that $T2$ of client $j$ finishes}
 
 $w_{j}=t+\ell_j$ \tcp*[f]{\small Set time that $T4$ of client $j$ can be processed}
 }
 \Else
 {Assign $T4$ of the client $j$ with the smallest index in $Q'$ such that $w_j\leq t$ to time intervals $S_t$ through $S_{t+p_{ij}'}$ on helper $i$.

 $Q'=Q'\setminus \{j\}$ \tcp*[f]{\small Update $Q'$ to remove $T4$ of client $j$}

 $t=t+p'_{ij}$ \tcp*[f]{\small Update $t$ to time that $T4$ of client $j$ finishes}

 $c_j:=t+r'_j$ \tcp*[f]{\small Set completion time of batch of client $j$}
 }
 }
 }
 \Return{$\max_{j\in\mathcal{J}} c_j$}
 \caption{\textbf{Algorithm~1}: $5$-approximation algorithm for \SL} \label{alg:approx}
\end{algorithm*}

Surprisingly, despite all of these negative results holding even for severely restricted instances of \SL, we manage to overcome this intractability by designing a polynomial-time 5-approximation algorithm for \SL.
As a subroutine in our 5-approximation algorithm (see line~1 in Algorithm~1), we employ the polynomial-time 2-approximation algorithm from~\cite{SS18} for the Generalized Assignment Problem with capacity constraints and without cost constraints ({\sc GAPcc} for short), which is formally defined below. We note that the polynomial-time 2-approximation algorithm from~\cite{SS18} considers {\sc GAPcc} with additional cost constraints, but these constraints are not needed for Algorithm~1, so we omit them from the problem definition.

\defoptproblem{{\sc GAPcc}}
{A bipartite graph $G=(\mathcal{J},\mathcal{I},\mathcal{E})$ composed of jobs~$\mathcal{J}$ and unrelated parallel machines~$\mathcal{I}$, positive integers $M_1,\ldots,M_{|I|}$, and a non-negative integer job processing time $p_{ij}^{\star}$ for all $i\in \mathcal{I}$ and $j\in \mathcal{J}$.}
{A schedule of minimum makespan where at most $M_i$ jobs are assigned to machine $i$ for all $i\in \mathcal{I}$ and the adjacency constraints of $G$ are respected, i.e., $j\in \mathcal{J}$ can be assigned to $i\in \mathcal{I}$ if and only if $(j,i)\in \mathcal{E}$.}

It is worth highlighting that the aforementioned polynomial-time 2-approximation algorithm for {\sc GAPcc} solves a fractional relaxation of the problem through linear programming, and then performs iterative rounding to obtain an integer solution. We can now establish our first main theoretical contribution.

\begin{theorem}\label{thm:approx}
     There exists a polynomial-time 5-approximation algorithm for \SL.
\end{theorem}

\begin{proof}
    We prove that Algorithm~1 is a polynomial-time 5-approximation algorithm for  \SL, and begin by proving its 5-approximation factor. We first prove that, when $r_j=\ell_j=r'_j=0$ for all $j\in \mathcal{J}$, the schedule produced by Algorithm~1 has at most the makespan of the schedule produced by the 2-approximation algorithm of~\cite{SS18}. 
    For an appropriate input, consider a schedule $S$ given by the 2-approximation algorithm of~\cite{SS18}. Consider a machine $i$ that, by $S$, is idle from time $t$ to time $t+x$ (where $t\geq 0$ and $x\geq 1$) and processes a job~$j$ from time $t+x$ to time $t+x+p_{ij}$, where $p_{ij}\geq 1$. Then, $S$ has the same makespan as the schedule $S'$ obtained from $S$ by instead processing job $j$ from time $t$ to time $t+p_{ij}$ on machine $i$ and having machine $i$ idle from time $t+p_{ij}$ to time $t+x+p_{ij}$. Hence, any schedule $S'$ will have a makespan that is at most that of $S$ if it (1) assigns the same jobs to the same machines as $S$, (2) processes them in the same order as~$S$, and (3) ensures that each machine is not idle until it has processed every job assigned to it. 
    
    Now, consider a machine $i$ that, by such a schedule $S'$, processes a job $j$ from time $t$ to time $t+p_{ij}$, and a job $j'$ from time $t+p_{ij}$ to time $t+p_{ij}+p_{ij'}$, where $t\geq 0$ and $p_{ij},p_{ij'}\geq 1$. Then, $S'$ has the same makespan as the schedule $S''$ obtained from $S'$ by instead processing job $j'$ from time $t$ to time $t+p_{ij'}$, and job $j$ from time $t+p_{ij'}$ to time $t+p_{ij'}+p_{ij}$ on machine $i$. Hence, any schedule $S''$ will have a makespan that is at most that of $S$ if it (1) assigns the same jobs to the same machines as $S$ and (2) ensures that each machine is not idle until it has processed every job assigned to it. Note that $S''$ can process the jobs at those machines in any order.   

    Let $\iota^{\star}$ be an instance of \SL. Consider an instance $\iota$ of \SL\ that is the same as $\iota^{\star}$ except that $r_j=\ell_j=r'_j=0$ for all $j\in \mathcal{J}$. Let OPT be the optimal value of the makespan for $\iota$. Due to the above, the schedule provided by Algorithm~1 for $\iota$ has a makespan of $k$ such that $k\leq 2\cdot \text{OPT}$ since it uses a feasible client-helper assignment $Y: \mathcal{J} \rightarrow \mathcal{I}$ provided by the 2-approximation algorithm of~\cite{SS18}. Indeed, the same arguments as in the paragraph above also prove that the tasks can be processed in any order that satisfies the precedence relations without increasing the makespan. 
     
    Let $\text{OPT}^{\star}$ be the optimal value of the makespan for~$\iota^{\star}$. Then, $\text{OPT}^{\star}\geq \text{OPT}$. Let $k^{\star}$ be the makespan given by Algorithm~1 for $\iota^{\star}$. Then, $k^{\star}\leq k+\max_{j\in \mathcal{J}}r_j+\max_{j\in \mathcal{J}}\ell_j+\max_{j\in \mathcal{J}}r'_j$, where the values of $\max_{j\in \mathcal{J}}r_j$, $\max_{j\in \mathcal{J}}\ell_j$, and $\max_{j\in \mathcal{J}}r'_j$ are with respect to $\iota^{\star}$. Indeed, in Algorithm~1, the maximum total amount of time a helper $i$ may remain idle while there is still a $T2$ to be processed on $i$ is at most $\max_{j\in \mathcal{J}}r_j$. Moreover, the maximum total amount of time a helper $i$ may remain idle while there is still a $T4$ to be processed on $i$ but no $T2$s to be processed on $i$ is at most $\max_{j\in\mathcal{J}}\ell_j$. Lastly, after the last $T4$ is processed on a helper~$i$, the maximum amount of time needed for the $T5$s of all the clients assigned to $i$ to be completed is at most $\max_{j\in \mathcal{J}}r'_j$.
    
    Observe that the following three inequalities trivially hold for $\iota^{\star}$: (a)~$\text{OPT}^{\star}\geq \max_{j\in \mathcal{J}}r_j$, (b)~$\text{OPT}^{\star}\geq \max_{j\in \mathcal{J}}\ell_j$, and (c)~$\text{OPT}^{\star}\geq \max_{j\in \mathcal{J}}r'_j$, where the values of $\max_{j\in \mathcal{J}}r_j$, $\max_{j\in \mathcal{J}}\ell_j$, and $\max_{j\in \mathcal{J}}r'_j$ are also with respect to $\iota^{\star}$. 
    Hence, $k^{\star}\leq k + 3\cdot \text{OPT}^{\star}\leq 2\cdot \text{OPT} + 3\cdot \text{OPT}^{\star}\leq 5\cdot \text{OPT}^{\star}$,~and~so, Algorithm~1 is a 5-approximation algorithm for \SL.

    Finally, we establish the desired running time bound for Algorithm~1. First, note that the 2-approximation algorithm of~\cite{SS18} runs in polynomial time, but using current methods requires at least $\mathcal{O}(n^{2+\varepsilon})$ time (for a constant $\varepsilon>0$ and input size $n$) due to the use of linear programming as a subroutine~\cite{JSWZ21}.    
    For the remaining steps in Algorithm~1, note that each of the FOR loops and the WHILE loop are iterated through at most $I$ or $J$ times. Thus, the overall running time of these loops is asymptotically upper-bounded by the time needed to perform the sorting steps described on lines~6-7. As these occur in the aforementioned FOR loop, the time required to perform all steps after line~1 can be upper-bounded by $\mathcal{O}(n^2\cdot \log n)$. Therefore, the overall asymptotic running time bound of the proposed algorithm matches the time required to invoke the known scheduling algorithm on line~1~\cite{SS18}.
\end{proof}

\section{\genSL: Hardness and Heuristic}\label{sec:genSL}

In this section, we study the computational complexity of the more general \genSL, and design a heuristic for it. Similarly to the fact that simply scheduling tasks at a single helper is strongly \NP-hard even for highly restricted instances (Theorem~\ref{thm:sched-hard2}), as our second main theoretical contribution, we prove that merely finding a feasible client-helper assignment for \genSL\ (formalized as the \assign\ problem) is strongly \NP-hard even for severely restricted instances. Hence, unless $\P=\NP$, \genSL\ cannot admit a polynomial-time approximation algorithm for any approximation factor.

\begin{theorem}\label{thm:assign-hard}
    Even when restricted to instances where $G$ is a complete bipartite graph, \assign\ is \textup{(1)} weakly \NP-hard even if $I=2$, \textup{(2)} \W\textup{[1]}-hard parameterized by $I$, and \textup{(3)} strongly \NP-hard even if $M_i<d_j+d_{j'}+d_{j''}+d_{j'''}$ for all $i\in \mathcal{I}$ and distinct $j,j',j'',j'''\in \mathcal{J}$ (i.e., at most 3 clients may be assigned to each helper).
\end{theorem}

\begin{proof}
    We prove the theorem via a polynomial-time reduction from $P~||~C_{\text{max}}$ to \assign. Given an instance $(\mathcal{J}^{\star},\mathcal{I}^{\star},k)$ of $P~||~C_{\text{max}}$, we construct an instance $G=(\mathcal{J},\mathcal{I},\mathcal{E})$ of \assign\ as follows. For all $i^{\star}\in \mathcal{I}^{\star}$, there is a helper $i\in \mathcal{I}$ with $M_i=k$. For all $j^{\star}\in \mathcal{J}^{\star}$, there is a client $j\in \mathcal{J}$ with $d_j=p_{j^{\star}}$. Finally, each client in $\mathcal{J}$ is adjacent to each helper in $\mathcal{I}$, i.e., $G$ is a complete bipartite graph. This completes the reduction, which clearly takes polynomial time.
    
    First, suppose that there is a schedule $S$ of makespan at most~$k$ for the instance $(\mathcal{J}^{\star},\mathcal{I}^{\star},k)$ of $P~||~C_{\text{max}}$. For all $j^{\star}\in \mathcal{J}^{\star}$, if job $j^{\star}$ is assigned to machine $i^{\star}\in \mathcal{I}^{\star}$ by $S$, then assign client $j$ to helper $i$. Let $Y: \mathcal{J} \rightarrow \mathcal{I}$ be the corresponding client-helper assignment. Since $G$ is a complete bipartite graph, we have that $Y$ respects the adjacency constraints of~$G$. As $S$ has makespan at most $k$ and the memory demands of the clients correspond to the processing times of the jobs (i.e., $d_j=p_{j^{\star}}$ for all $j\in \mathcal{J}$), then $\sum_{j\in Z_Y(i)}d_j\leq k=M_i$ for all $i\in \mathcal{I}$. Hence, $Y$ also respects the memory constraints, and thus, $Y$ is a feasible client-helper assignment.

    For the reverse direction, suppose there is a feasible client-helper assignment $Y: \mathcal{J} \rightarrow \mathcal{I}$ for the instance $G=(\mathcal{J},\mathcal{I},\mathcal{E})$ of \assign. For all $i\in \mathcal{I}$ and $j\in Z_Y(i)$, assign job $j^{\star}$ to machine $i^{\star}$, and let $Y^{\star}: \mathcal{J}^{\star} \rightarrow \mathcal{I}^{\star}$ be the corresponding assignment. For all $i^{\star}\in \mathcal{I}^{\star}$, schedule the jobs in $Z_{Y^{\star}}(i^{\star})$ on machine $i^{\star}$ in any order, maintaining that machine $i^{\star}$ is not idle until it has processed every job in $Z_{Y^{\star}}(i^{\star})$. Let~$S$ denote the corresponding schedule. 
    As~$Y$ is a feasible client-helper assignment, then $\sum_{j\in Z_Y(i)}d_j\leq M_i=k$ for all $i\in \mathcal{I}$.
    As the memory costs of the clients correspond to the processing times of the jobs (i.e., $d_j=p_{j^{\star}}$ for all $j\in \mathcal{J}$), and $S$ ensures that the machines are not idle while they have not processed every job assigned to them, then $\sum_{j^{\star}\in Z_{Y^{\star}}(i^{\star})}p_{j^{\star}}\leq M_i=k$ for all $i^{\star}\in \mathcal{I}^{\star}$. Thus, $S$ is a schedule of makespan at most $k$.
    
    The theorem statement then holds as $P~||~C_{\text{max}}$ is (1) weakly \NP-hard even if $\mathcal{I}^{\star}=2$~\cite{GJ79}, (2) \W[1]-hard parameterized by the number of machines~\cite{JKMS13,KK18}, and (3) strongly \NP-hard even if $k/4<p_{j^{\star}}<k/2$ for all $j^{\star}\in \mathcal{J}^{\star}$ (this particular case of $P~||~C_{\text{max}}$ is known as the {\sc 3-Partition} problem)~\cite{GJ79}.
\end{proof}

Due to this negative result, the best that can be hoped for is a heuristic. Building on our approximation algorithm for \SL, we design a heuristic for \genSL, where the only difference between the two algorithms is in the first step where the client-helper assignment is computed (i.e., line~1 in Algorithm~1). This step must be adapted as it is strongly \NP-hard to even compute a feasible client-helper assignment for \genSL\ (Theorem~\ref{thm:assign-hard}), i.e., it cannot be done in polynomial time unless $\P=\NP$.
Thus, for the heuristic, we instead use a solver to find a feasible client-helper assignment that minimizes the maximum sum of the T2s and T4s assigned to a helper over all helpers, i.e., that finds a client-helper assignment $Y$ that is a solution to the problem $\min_Y{\max_{i\in \mathcal{I}}{\sum_{j\in Z_Y(i)}p_{ij}+p'_{ij}}}$ subject to $Y$ being feasible. In other words, our heuristic equally distributes the sums of the T2s and T4s of the clients assigned to the helpers, and hence, we naturally call it \textbf{EquiDistributed} (\textbf{EquiD} for short). In the next section, the experimental evaluations of EquiD show that this equal distribution
is the key to obtaining an efficient heuristic that also performs well.

\section{Numerical Evaluations}\label{sec:eval}

In this section, we evaluate the performance of EquiDistributed (EquiD) for \genSL, and compare it with baselines and methods from the literature. Our experiments show that EquiD often yields near-optimal outputs and runs in a fraction of the time a solver takes to give an optimal solution. We begin with the experimental setup before moving on to the insights gained from the empirical analysis.

\subsection{Experimental Setup}
\noindent \textbf{Dataset \& Heterogeneity Levels.} For our evaluations, we~use open-source data of time measurements and memory requirements per layer and per batch from~\cite{joana_github}. This data concerns the training of ResNet101~\cite{he2016deep} and VGG19~\cite{simonyan2014very} with two different datasets~(CIFAR-10~\cite{krizhevsky2009learning} and MNIST~\cite{lecun2002gradient}) by 6 different heterogeneous devices: laptop, virtual machine, RPi3, RPi4, Jetson GPU, and Jetson CPU~(for MNIST, measurements are available for only the 4 first devices). The first two serve as helpers and the rest are clients.  The data in~\cite{joana_github} has similar characteristics to other time measurements from SL testbeds in the literature, e.g.,~\cite{wang2023coopfl, hafi2024impact}, in terms of disparity of processing times across layers and across forward/backward propagation.

We create problem instances from the data described above, but we also create synthetic data. The instances have four different levels of heterogeneity. Level~1 considers only 2 different types of client devices with the same cut layers, and the 2 types of helper devices. As the cut layers and type of device determine the lengths of tasks T1 through~T5, this is a rather homogeneous scenario. Level~2 considers all the different client devices with the same cut layers, and 2 types of helper devices. Level 3 is the same as level 2, except that the cut layers are selected randomly~(where the first cut layer is in the first few layers of the NN and the second cut layer is in the last few). Level 4 is similar to level~3, but uses synthetic data for measurements: they are selected randomly, but in the range of the actual data. In levels 1-3, the memory demands and capacities are based on the data, while for level 4 they are selected uniformly at random but lie within the existing range. Unless otherwise specified, the speed of the connection between the devices and the cloud~(helpers), which together with the memory demands determine the values of $r_j, \ell_j$, and $r_j'$, follows the statistics in Akamai's Internet Report~\cite{belson2017q4}. 

\noindent \textbf{Metrics \& Comparison with Other Methods.} We are mainly interested in the training makespan metric (for a batch), but we also analyze the execution times of our algorithm EquiD.\footnote{We stress here that metrics like accuracy are out of the scope of this work since the assignments and scheduling decisions do not affect the convergence of the model. Future work could explore how client dropouts or client selection policies could impact the makespan or the model's convergence.} In~particular, we compare the performance of EquiD (in terms of the computed makespan) to the following methods:
\begin{itemize}
    \item \textbf{B-G~(balanced-greedy)}: this is the method proposed in~\cite{tirana2024workflow}. For each client, it finds the helper with the smallest number of assigned clients and enough available memory, and assigns it to this helper. Then, the schedule follows the first-come-first-serve~(FCFS) policy.
    \item \textbf{ED-FCFS}: this is a baseline bridging EquiD and B-G. It solves the assignment problem in the same way as EquiD, but then it schedules the clients in the FCFS order.
\end{itemize}
Finally, a comparison with other methods in the literature~(e.g., those mentioned in Subsection~\ref{subsec:related}) is not possible since these jointly decide on other aspects of the general SL framework such as radio spectrum allocation or device clustering.

\subsection{Insights}
First, we compare EquiD with the optimal solution of \genSL. Table~\ref{table:optimal} reports the suboptimality for various instance sizes~(in terms of $J$ and $I$) and heterogeneity levels for ResNet101 and CIFAR10. Due to the time it takes the solver to find an optimal solution, only small instances of the problem can be solved optimally.\footnote{We use the mathematical model from~\cite{tirana2024workflow} with a time slot of 300ms.} For this reason, the execution times are also provided. Here, the optimal's execution time is the one achieved by the solver Gurobi~\cite{gurobi}, while we observed mostly similar times with the SCIP solver~\cite{bolusani2024scip}. Both solvers were run on an off-the-shelf laptop. As the heterogeneity increases, the solver needs up to $683$ seconds to find the optimal solution. Overall, the following observation can be extracted from the data in Table~\ref{table:optimal}.

\smallskip

\noindent \emph{Our algorithm EquiD typically achieves a suboptimality of at most $7.79\%$ (and at most $19.77\%$ in the worst case) and runs in a fraction of the time a solver needs to solve \genSL\ optimally.}

\smallskip

\begin{figure*}[!t]
\centering  
\includegraphics[width=18cm, trim={0.4cm 0.4cm 0.2cm 0.1cm},clip]{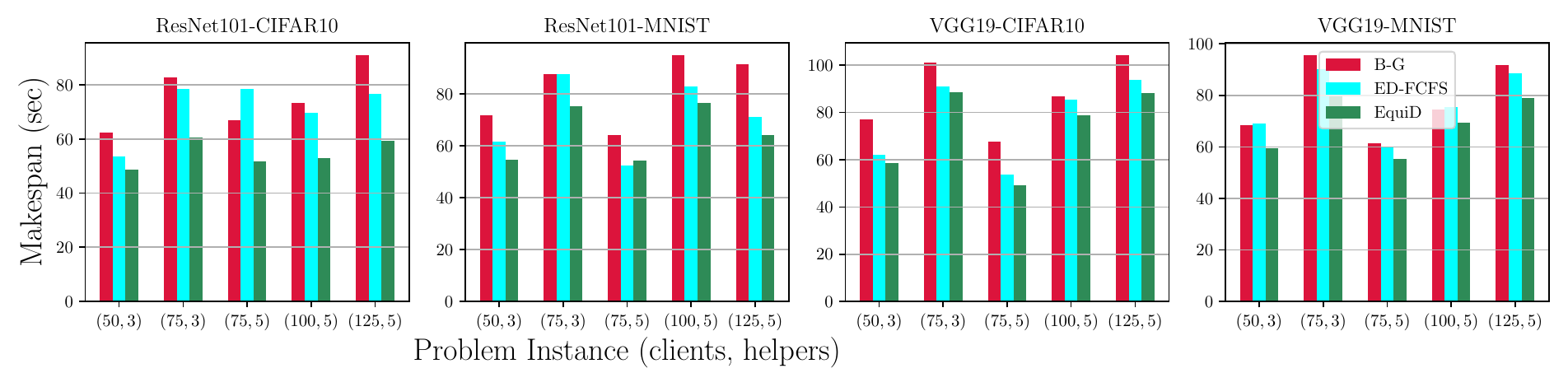} 
\caption{Batch makespan (in sec) for different problem instances achieved by our algorithm EquiD, the baseline ED-FCFS, and B-G, the algorithm from~\cite{tirana2024workflow}. Note that our algorithm EquiD computes a smaller makespan than the other two methods in every scenario except one (ResNet101-MNIST with 75 clients and 5 helpers), where only ED-FCFS computes a slightly smaller makespan.}  
\label{fig:makespan}
\end{figure*}

\begin{table}
\begin{center}
\caption{Comparison of the proposed policy EquiD with the optimal solution in terms of suboptimality and execution time.} \label{table:optimal}
\begin{tabular}{ |c|c|c|c|c|c| } 
\hline
\multirow{1.5}{*}{Heterogeneity} & \multirow{1.5}{*}{$J$} & \multirow{1.5}{*}{$I$} & \multirow{1.5}{*}{Suboptimality} & \multirow{1.5}{*}{Optimal's ex.}  & \multirow{1.5}{*}{EquiD ex.} \\[3pt] level & & & (\%) & time (sec) & time (sec)
\\[2pt]
\hline
 \multirow{9}{*}{2} & \multirow{1.5}{*}{8} & \multirow{1.5}{*}{2} & \multirow{1.5}{*}{6.15} & \multirow{1.5}{*}{63.52} & \multirow{1.5}{*}{0.02}
\\[3pt] \cline{2-6} & \multirow{1.5}{*}{10} & \multirow{1.5}{*}{2} & \multirow{1.5}{*}{0.00} & \multirow{1.5}{*}{42.82} & \multirow{1.5}{*}{0.02} 
\\[3pt] \cline{2-6} & \multirow{1.5}{*}{10} & \multirow{1.5}{*}{5} & \multirow{1.5}{*}{1.67} & \multirow{1.5}{*}{140.07} & \multirow{1.5}{*}{0.03}
\\[3pt] \cline{2-6} & \multirow{1.5}{*}{12} & \multirow{1.5}{*}{2} & \multirow{1.5}{*}{7.79} & \multirow{1.5}{*}{94.07} & \multirow{1.5}{*}{0.02}
\\[3pt] \cline{2-6} & \multirow{1.5}{*}{15} & \multirow{1.5}{*}{2} & \multirow{1.5}{*}{19.77} & \multirow{1.5}{*}{146.65} & \multirow{1.5}{*}{0.02}
\\[3pt] \cline{2-6} & \multirow{1.5}{*}{15} & \multirow{1.5}{*}{5} & \multirow{1.5}{*}{2.78} & \multirow{1.5}{*}{449.84} & \multirow{1.5}{*}{0.02}
\\[3pt]
\hline
 \multirow{9}{*}{3} & \multirow{1.5}{*}{8} & \multirow{1.5}{*}{2} & \multirow{1.5}{*}{2.10} & \multirow{1.5}{*}{71.21} & \multirow{1.5}{*}{0.01}
\\[3pt] \cline{2-6} & \multirow{1.5}{*}{10} & \multirow{1.5}{*}{2} & \multirow{1.5}{*}{0.56} & \multirow{1.5}{*}{60.11} & \multirow{1.5}{*}{0.01} 
\\[3pt] \cline{2-6} & \multirow{1.5}{*}{10} & \multirow{1.5}{*}{5} & \multirow{1.5}{*}{2.86} & \multirow{1.5}{*}{136.54} & \multirow{1.5}{*}{0.03}
\\[3pt] \cline{2-6} & \multirow{1.5}{*}{12} & \multirow{1.5}{*}{2} & \multirow{1.5}{*}{1.94} & \multirow{1.5}{*}{156.77} & \multirow{1.5}{*}{0.02}
\\[3pt] \cline{2-6} & \multirow{1.5}{*}{15} & \multirow{1.5}{*}{2} & \multirow{1.5}{*}{0.00} & \multirow{1.5}{*}{241.97} & \multirow{1.5}{*}{0.01}
\\[3pt] \cline{2-6} & \multirow{1.5}{*}{15} & \multirow{1.5}{*}{5} & \multirow{1.5}{*}{2.50} & \multirow{1.5}{*}{683.34} & \multirow{1.5}{*}{0.02}
\\[3pt]
\hline
\end{tabular}
\end{center}
\end{table}

Next, we focus on how EquiD performs compared to other methods. Figure~\ref{fig:makespan} depicts the makespan of EquiD, B-G, and ED-FCFS for different problem instances.

\smallskip

\noindent \emph{EquiD consistently outperforms the methods B-G and ED-FCFS by achieving a shorter makespan by up to $34.6\%$~(in the case of ResNet101 and CIFAR10).}

\smallskip

In most cases, ED-FCFS performs better than B-G, but worse than EquiD. This highlights how different client-helper assignments may affect the makespan, and how important it is that these assignments are decided while taking into account the processing times of tasks $T2$ and $T4$. Further, the performance gap between ED-FCFS and EquiD derives from the different scheduling policies. An FCFS schedule that is oblivious to the lengths of T3 and T5 might not prioritize a slow client~(straggler), and thus, lead to a larger makespan.

In Figure~\ref{fig:makespan}, we see that
the performance gaps are smaller in the case of VGG19 when compared to ResNet101. This is because  the size of activations/gradients that need to be exchanged when training VGG19 are on average larger  than the ones of ResNet101. Thus, depending on the cut layers, communication delays may be large and some clients may be very slow, making the makespan largely independent of the optimization decisions. For this reason, for the experiments carried out for VGG19 in Figure~2, the clients' connectivities were selected to be within the fastest connectivity range.

Next, we further compare EquiD with B-G while focusing on the impact of task heterogeneity and the number of clients and helpers. Figure~\ref{fig:relative} reports how much larger the makespan given by B-G is relative to the makespan given by EquiD for different heterogeneity levels for ResNet101 and CIFAR10. The relative difference increases as the number of clients and helpers increases~(i.e., when comparing the pink and gray bars). This is due to the fact that, in the presence of a larger number of clients, assignments and scheduling decisions need to prioritize the slowest devices. That is precisely what EquiD does by scheduling not only based on the lengths of the T2s and T4s, but also the T3s and T5s (lines 6-7 in Algorithm~1).

Looking now at the heterogeneity levels in Figure~\ref{fig:relative}, while for low levels the relative difference can be as low as $25\%$, it drastically increases in highly heterogeneous scenarios.

\smallskip

\noindent \emph{In highly heterogeneous scenarios, B-G gives a $70.4\%$ to $117.8\%$ larger makespan than the one given by EquiD.}

\smallskip

To complete the comparison of EquiD with B-G, it is crucial to note that EquiD is guaranteed to provide a solution if one exists, i.e., a schedule adhering to the constraints of \genSL. However, this is not true for B-G since it can fail to find a feasible client-helper assignment. A simple example demonstrating this is the case of two helpers with memory capacities $1$ and $2$, and two clients with memory demands $1$ and $2$. While it is clear that a feasible client-helper assignment exists in this case, the B-G policy may assign the client with memory demand $1$ to the helper with memory capacity $2$, in which case the other client cannot be assigned to any helper.

Finally, in Figure~\ref{fig:helpers}, we study the performance of EquiD as the number of clients and helpers vary under high heterogeneity~(i.e., level 4) for ResNet101 and MNIST. As expected, the makespan increases as the number of clients increases. This is also the case when the number of helpers decreases. An exception to this observation occurs in the cases where there are 50 and 75 clients. In these cases, the makespan is determined by some slow devices, and although they are assigned to a fast helper, their training time is not affected by the addition of helpers. This is also the reason why, for 5 helpers, the makespan remains stable when there are at least 75 clients. As our final observation, we mention the following. 

\smallskip

\noindent \emph{As the number of helpers increases, the makespan decreases by up to $21.8\%$. This phenomenon becomes more prominent when the number of clients is large.}

\begin{figure}[!t]
\centering  
\includegraphics[width=0.886\columnwidth, trim={0cm 0.33cm 0.2cm 0.3cm},clip]{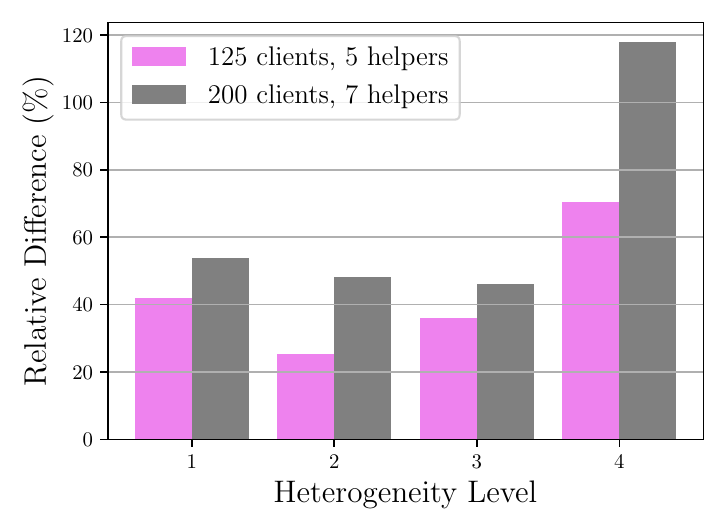} 
\caption{Relative difference in terms of how much larger the makespan given by B-G is than the makespan given by EquiD for different levels of heterogeneity. Here, the time measurements of ResNet101 and CIFAR10 were used.}  
\label{fig:relative}
\end{figure}

\begin{figure}[!t]
\centering  
\includegraphics[width=0.886\columnwidth, trim={0.2cm 0.33cm 0.2cm 0cm},clip]{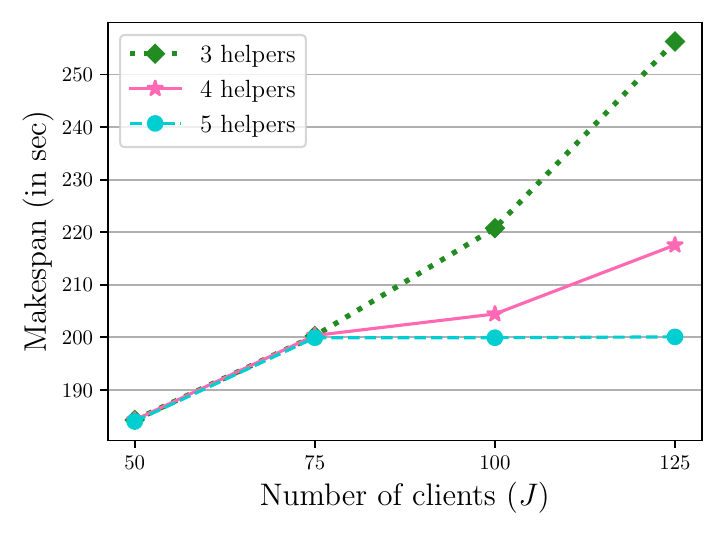} 
\caption{Makespan achieved by EquiD as the number of clients and helpers varies (for ResNet101 and MNIST).}  
\label{fig:helpers}
\end{figure}

\section{Conclusion and Future Work}

In this work, we studied the problem of minimizing the makespan in SL for homogeneous and heterogeneous task scenarios, captured by the memory constraints in \SL\ and \genSL, respectively. Even for highly restricted instances, we proved that \SL\ is strongly \NP-hard  and cannot admit a polynomial-time (3/2)-approximation algorithm unless~$\P=\NP$. Nevertheless, we overcame this intractability by providing a polynomial-time 5-approximation algorithm for \SL. For \genSL, we established that merely determining whether there exists a feasible client-helper assignment is strongly \NP-hard, implying that it cannot admit a polynomial-time approximation algorithm for any approximation factor. Nonetheless, we designed a heuristic called EquiD based on our approximation algorithm, and showed through extensive experiments that it outperforms the other heuristics from the literature and provides near-optimal solutions in a fraction of the time it takes for a solver to find the optimal solution.

Future work would naturally involve including further decisions in the SL process such as deciding on the cut layers. However, as we have shown that highly restricted instances of our problems are already hard, most likely only heuristics would be attainable in this case. One of the most interesting future directions is to find the best approximation factor possible by a polynomial-time approximation algorithm for \SL. Since our lower bound holds for its restriction to $R~||~C_{\text{max}}$, it seems plausible that the additional features of \SL\ could be used to increase this bound.

%%% LAST PAGE FOR  REFERENCES
\clearpage
{
\bibliography{SL-bib}
\bibliographystyle{ieeetr}
}

\end{document}